\documentclass[10pt,journal,  oneside, twocolumn]{IEEEtran}

\usepackage{multirow}
\usepackage{latexsym}
\usepackage{graphicx}
\usepackage{float}
\usepackage{amsmath}
\usepackage{amsthm}
\usepackage{lipsum}
\usepackage{subfig}
\usepackage{graphicx}
\usepackage{authblk}
\usepackage{bm}
\usepackage{amsthm}
\usepackage[section]{placeins}
\usepackage{wasysym}

\usepackage{colortbl}

\usepackage{enumitem}

\newtheorem{theorem}{Theorem}

\newtheorem{lemma}{Lemma}

\makeatletter  
\newif\if@restonecol  
\makeatother

\usepackage[linesnumbered,ruled,vlined]{algorithm2e}
\usepackage{algpseudocode}  
\usepackage{amsmath}

\usepackage{amssymb}

\usepackage{mathrsfs}
\usepackage{subfig}
\usepackage{caption}
\begin{document}

\title{RIS-empowered Topology Control for \\Distributed Learning in Urban Air Mobility}


\author{Kai Xiong, Rui Wang, Supeng Leng,~\IEEEmembership{Member,~IEEE}, Wenyang Che, Chongwen Huang,~\IEEEmembership{Member,~IEEE}, \\Chau Yuen,~\IEEEmembership{Fellow,~IEEE}

\thanks{

K. Xiong,  R. Wang, S. Leng, and W. Che are with School of Information and Communication Engineering, University of Electronic Science and Technology of China, Chengdu, 611731, China; and, Shenzhen Institute for Advanced Study, University of Electronic Science and Technology of China, Shenzhen, 518110, China.
}

\thanks{
C. Huang is with College of Information Science and Electronic Engineering, Zhejiang University, Hangzhou 310027, China.
}

\thanks{
C. Yuen is with School of Electrical and Electronics Engineering, Nanyang Technological University, 639798, Singapore. 
}






\thanks{The corresponding author is Supeng Leng, email: \{spleng, xiongkai\}@uestc.edu.cn}
}


\maketitle


\begin{abstract}
Urban Air Mobility (UAM) expands vehicles from the ground to the near-ground space, envisioned as a revolution for transportation systems. 
Comprehensive scene perception is the foundation for autonomous aerial driving.
However, UAM encounters the intelligent perception challenge: high perception learning requirements conflict with the limited sensors and computing chips of flying cars. 
To overcome the challenge, federated learning (FL) and other collaborative learning have been proposed to enable resource-limited devices to conduct onboard deep learning (DL) collaboratively.
But traditional collaborative learning like FL relies on a central integrator for DL model aggregation, which is difficult to deploy in dynamic environments.
The fully decentralized learning schemes may be the intuitive solution while the convergence of distributed learning cannot be guaranteed.
Accordingly, this paper explores reconfigurable intelligent surfaces (RIS) empowered distributed learning, taking account of topological attributes to facilitate the learning performance with convergence guarantee.
We propose several FL topological criteria for optimizing the transmission delay and convergence rate by exploiting the Laplacian matrix eigenvalues of the communication network. 
Subsequently, we innovatively leverage the RIS link modification ability to remold the current network according to the proposed topological criteria. 
This paper rethinks the functions of RIS from the perspective of the network layer.
Furthermore, a deep deterministic policy gradient-based RIS phase shift control algorithm is developed to construct or deconstruct the network links simultaneously to reshape the communication network. Simulation experiments are conducted over MobileNet-based multi-view learning to verify the efficiency of the distributed FL framework.

\end{abstract}


\IEEEpeerreviewmaketitle

\section{Introduction}

\IEEEPARstart {R}{oad} traffic congestion remains a major hazard that causes heavy damage to the economy and life quality of the metropolises. 
Urban Air Mobility (UAM) is a future aerial transportation system for alleviating ground traffic congestion, counting on flying car technology. 
National Aeronautics and Space Administration (NASA) pointed out that the UAM vehicle should safely and intelligently aviate in metropolitan areas.
Wherein flight safety stands as the top priority that forms the foundation for any design of flying cars \cite{Ehang3248}. 
Due to monitoring the environment being a prerequisite for safe aerial driving, the essential function of the flying car prototype design involves intelligent perception.

Generally, comprehensive scene perception requires considerable multi-view data and computing resources. 
However, limited onboard sensors and chips restrict the perception capacity of a vehicle.
An alternative way is to achieve collaborative perception with neighbors, enhancing vehicular data collection and computing capability.
But urban air communication often encounters blocking caused by skyscrapers and experiences high path loss.
Reconfigurable Intelligent Surface (RIS) has recently emerged as an obstacle communication solution for future sixth-generation wireless communication systems.
It can reflect the signals toward users who experience direct link blockages due to obstacles \cite{9831036Alsenwi}.
Besides, RIS has a compact size and low power consumption due to the reflective trait and straightforward structure. 
It suitably mounts on the building facades and enables resource-limited devices, like flying cars and unmanned aerial vehicle \cite{10278101Hongyang}, to promote the information exchange with each other \cite{{9599478Neelima},{9899368Yaxuan},{10041955Nouman},{10032196Hongyang}}.


An emerging collaborative learning approach, Federated Learning (FL), serves resource-limited devices without sharing the raw sensing data with others.
However, the central integrator is another difficulty that hinders the application of FL in UAM.
Maintaining an integrator in the air is impractical for a highly dynamic scenario \cite{9079513Shiva}.
The existing central integrator may also invoke a network bottleneck, significantly raising the communication overheads \cite{9716076Afaf}.
One alternative solution is to perform the deep learning (DL) model aggregation fully distributed without the integrator involved, where flying cars share their local DL model parameters with communication neighbors.
Nevertheless, the shared parameter consensus is a fundamental provision for distributed learning optimization \cite{{9354489Xindi}, {9525810Zhao}}.
Inconsistent parameters among distributed participants will corrupt the distributed learning performance.
Therefore, platoon vehicles must agree upon the shared DL model parameters with neighbors in a distributed fashion.

Additionally, many studies have revealed that the end-to-end transmission delay and communication topology have a crucial impact on the consensus of multi-agent systems \cite{{8568992Qian}, {1470321Saber},{Saber1333204}}. 
Transmission delays and inappropriate topology may deteriorate the DL model propagation efficiency or render the system erratic.
Although the distributed FL inevitably faces delay-shift parameters and intricate convergence in wireless communication systems, retrospectives have yet to explore the topological solution to improve the DL model sharing in distributed learning.
As a result, it is a research blank to develop DL model parameter consensus and networking configurations for distributed learning.

Fortunately, RIS with massive reflective elements can tune the phase shifts of the impinging electromagnetic waves \cite{{10183797Bo}, {9722711Cheng}}. 
By programming the phase shifts of RIS elements, the reflected signals can be composited constructively at the desired receiver or destructively at non-intended receivers \cite{9903846Hehao}.
We can manipulate the RIS to construct and deconstruct the reflected signals to reshape the communication topology and achieve the desired communication rate for distributed learning.


Thus, this paper proposes a RIS-based topological optimization scheme to improve the convergence and learning performance of distributed FL.
Hereafter, we explore the MobileNet-based multi-view learning to evaluate the proposed distributed FL framework.
The multi-view learning with distributed FL framework can deploy fully distributedly and fast converge with loose transmission delay constraints.
Note that the above distributed FL and RIS-based topology optimization scheme can directly deploy on multi-agent systems to facilitate swarm intelligence.
The main contributions are summarized as follows:

\begin{itemize}

\item We propose several topological criteria to ensure distributed learning convergence and relieve the transmission delay requirement of the convergence. These criteria aim to minimize the largest eigenvalue and maximize the second smallest eigenvalue of the Laplacian matrix of the communication network. Wherein the largest eigenvalue confines the tolerable transmission delay for convergence, which is constrained by the maximum vertex degree of the network. Instead, the second smallest eigenvalue specifies the convergence speed of the distributed learning. 
Further, we originally proposed a neighbor structure to increase the second smallest eigenvalue. These proposed criteria are simple and straightly deployed in practice.


\item We design a RIS-empowered topology optimization scheme by considering the transmission capacity and the topological criteria to reshape the communication network of a platoon.
The proposed scheme construct and deconstruct the communication link by accommodating the transmission rate.
Based on the Laplacian matrix analysis, we derive the tolerable transmission delay for the convergence of distributed FL.
The tolerable delay specifies the transmission rate threshold of the constructed link, directing the RIS controls.

\item We develop a deep deterministic policy gradient (DDPG) method to optimize the RIS phase shift matrix for distributed learning.
The reward of the DDPG comprises the transmission bonus and the delay violation penalty. 
The delay violation penalty incentivizes each link to observe the tolerable delay threshold, which enhances the transmission rate of the constructed links and suppresses the deconstructed links. 
With increased RIS reflective elements, the DDPG method exhibits global optimization ability to attain the desired transmission rates of each vehicle.

\end{itemize}

The remainder of this paper is organized as follows. Section II reviews the related works. Section III presents the whole system architecture and proposes the RIS-empowered topology optimization scheme. Section IV provides the simulation results and the performance discussion. Finally, we draw the conclusion in Section V. 

\section{Related Work}

Engineers and researchers have begun proposing, examining, and developing flying cars for efficient transport.
Autonomous driving is crucial in flying car design, which relies on comprehensive onboard perception. 
Federated learning can further improve the onboard perception performance.
Nevertheless, federated learning highly relies on the communication topology and corresponding transmission delays.

\subsection{Flying Car}
Flying cars can run as regular vehicles on the road and aviate as Vertical TakeOff and Landing (VTOL) aircraft in the air \cite{7420801Kaushik}.
With governments, enterprises, and research institutes paying increasing attention to UAM, this new concept has developed fast.
A blue paper by Morgan Stanley estimates that the global UAM addressable market will reach US dollar $1.5$ trillion by 2040, which is on the same scale as the potential market of the autonomous vehicle \cite{Morgan2019}.  
Flying cars are expected to transit an extensive area at high speed according to passenger requests, which incurs a flying car experiencing diverse environments in its life cycle.
In addition, safe aerial autonomous driving relies on comprehensive surrounding perception.
A platoon of flying cars enables large-scale distributed perception learning with high precision and enlarges the perception knowledge by cooperation \cite{7452570Fang}. 
It indicates that collaborative learning can facilitate perception accuracy by integrating different trained DL models.


\subsection{Distributed Federated Learning}
Multi-agent collaborative learning has been widely investigated in transportation systems \cite{{9127823Chai},{10149027Xiaosha}}. 
A central integrator can merge learning models from multiple locally trained DL models to prompt the learning efficiency. 
This collaborative learning fashion is named Federated Learning (FL) \cite{9084352Tian}.
The authors in \cite{9181432kai} studied FL-based distributed communication technology selection in Internet-of-Vehicles.
The necessity of a central integrator in FL incurs new challenges to the participant locations, the integrator's stabilization, and the capacity of the transmission bottleneck \cite{{9802876Tianyang}, {9839387Longyu}}. The intuitive solution is to decentralize the DL model-sharing process.
However, the convergence of the fully distributed FL is not guaranteed. 
Many works have exposed that the communication topology and transmission delay are the crucial factors affecting the convergence of the multi-agent parameter sharing \cite{{1470321Saber},{Saber1333204}}. But how to devise a communication topology to prompt the distributed FL convergence and performance is still not addressed.

\subsection{RIS transmission control}
RIS is a meta-surface with integrated electronic circuits that can alter an incoming electromagnetic field in constructing or deconstructing manners \cite{9410457Chongwen}.
The authors in \cite{9110869Chongwen} joint optimized the transmission beamforming of the base station and the RIS through the DDPG method.
An Unmanned Aerial Vehicle (UAV) equipped with the RIS has been investigated in \cite{9416239Sixian} that functions as a passive relay to minimize the personal data rate received by the eavesdropper.
Although previous literature thoroughly researched the RIS's signal construction/deconstruction capability, it remains in the physical layer to explore the RIS functions \cite{{9899454Xu},{9500188Xu}}.
Typically, the link deconstruction function of RIS retrospectives is only devoted to secure communications \cite{{9446916Jingyi}, {10175566Shengbin}}.
However, the multiple RIS constructions and deconstructions will inevitably influence the topology and routing behaviors in the network layer.
Nevertheless, the multiple RIS constructions and deconstructions will inevitably influence the topology and routing behaviors in the network layer.
There is a research gap in thinking at the network layer to utilize the RIS, remolding the communication network. \\

Overall, how to devise a suitable network for distributed FL upon resource-limited flying cars is a blank in retrospect. 
This paper aims to develop a RIS-empowered topology control framework to improve the onboard learning ability of flying cars. The RIS-based communication topology control yields a ubiquitous DL model sharing with distributed fashion. 
This paper first leverages the RIS construction and deconstruction abilities to remold the communication network. 
It is also appropriate for the UAV swarm and multi-agent systems.

\section{RIS-empowered Distributed FL framework}
The core goal of this paper is to design a novel RIS-empowered communication topology control for efficient distributed FL in UAM, taking the convergence of DL model propagation into account.
The proposed RIS-empowered distributed FL scheme consists of two phases, i.e., topology optimization and distributed FL implementation.
In the topology optimization phase, we propose several topology criteria to generate an optimal communication topology for fast convergence of distributed DL models and relieve transmission delay requirements.
This paper utilizes the RIS link construction and deconstruction abilities to tailor the communication network.
In the distributed FL implementation phase, flying cars employ the RIS-optimized topology to propagate their shared DL models for DL model aggregation.
After several DL model sharings, the distributed FL yields the final converged outcomes.

Instead of a centralized aggregation by the integrator, the proposed distributed FL adopts a distributed aggregation protocol that weight-averages the received DL model parameters from neighbors in a distributed manner.
Hence, the vehicles only participate in cooperative training with neighbors rather than the central integrator.
It indicates that the availability of neighbors drastically impacts the vehicles' perception performance, which stems from the communication topology of the system.

\begin{figure}[h]
\centering
     \includegraphics[width=.44\textwidth]{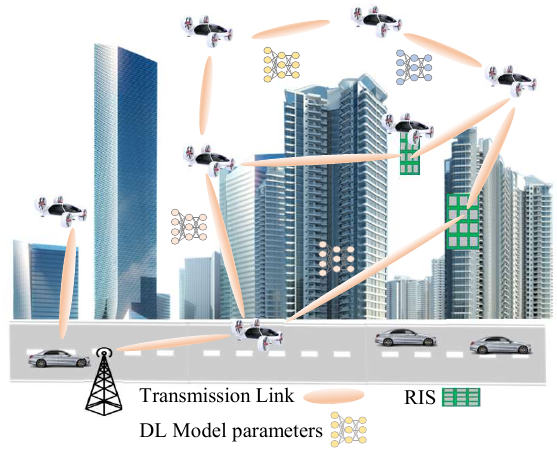} 
     \caption{RIS-empowered distributed FL in UAM scenario. } 
\label{paper8scenario}
\end{figure}

Thus, this paper concentrates on the distributed FL design and the corresponding topology optimization.
Fig.~\ref{paper8scenario} illustrates the RIS-empowered distributed FL implementation in the UAM scenario.
We consider an Internet-of-Vehicle of several spatially dispersed flying cars equipped with one single antenna.
Flying cars transmit their information or DL model parameters with communication neighbors through RIS or direct channels. 
The RIS can be mounted on building surfaces or flying cars. It could efficiently manipulate the propagation environment with multiple reflected elements. 

\subsection{Topology Design for Distributed FL}
We first introduce the distributed federated learning scheme. 
The distributed FL highly relies on the neighbor availability.
Suppose that each flying car communicates with its neighbors $N_i = \{j \in V| \{i,j\} \in  E \}$ on the communication network $G=(V,E)$, wherein $V=\{1,2,\dots m\}$ is the set of flying cars. 
Denote by $E\in\{V,V\}$ the set of communication links.
Then, we propose a distributed aggregation protocol between vehicle $i$ and its neighbor $j$ for the distributed FL carried out on the network $G$, which is given as,
\begin{equation}
\begin{aligned}
& v_i (k+1, t) =\\ &\sum_{v_j \in N_i}  a_{ij}   \left[v_j (k, t-\tau_{ij} )-v_i (k, t-\tau_{ij} )\right] +v_i (k, t),
\end{aligned}
\label{fjfjfjfjkkleie}
\end{equation} 
\noindent where $t$ represents the time variable and $k$ is the iteration epochs. 
$\tau_{ij}$ is the DL model transmission delay from the car $j$ to car $i$. 
$a_{ij}$ is the weight of the edge $(j,i)$. 
To simplify the analysis, we regard $a_{ij}=1$ as the direct connection between $i$ and $j$. 
Otherwise, $a_{ij}=0$ indicates no direct link between $i$ and $j$. 
Based on the previous work \cite{Saber1333204}, a necessary and sufficient condition for convergence of the distributed aggregation is, 
\begin{equation}
\begin{aligned}
\tau_{max} \le \frac{\pi}{2 \lambda_{max}},
\end{aligned}
\label{sdajhgfduireawuihfcvhj}
\end{equation} 

\noindent where $\tau_{max}$ is the maximum tolerable transmission delay of all communication links $E$. 
And $\lambda_{max}$ is the largest eigenvalue of the Laplacian matrix $L(G)$ of network $G$. 
The Laplacian matrix is expressed as $L(G)=D(G)-A(G)$ where $D(G)$ is a diagonal degree matrix of network $G$ with the $i$-th diagonal element $d_i=|N_i|$, and the non-diagonal elements of $D(G)$ are zero. 
And, $A(G)$ is an adjacency matrix composed by $a_{ij}$. 
Thus, the element of $L$ is
\begin{equation}
\begin{split}
\begin{aligned}
l_{i j}= \begin{cases}-1, & j \in N_i \\ \left|N_{i}\right|, & j=i\end{cases}.
\end{aligned}
\end{split}
\label{ujjjjrycsdasd}
\end{equation}

\noindent When the delay constraint (\ref{sdajhgfduireawuihfcvhj}) holds, all shared DL model parameters will asymptotically converge to the average value $\bar{v}=\frac{1}{| N_i |} \sum_{i} v_{i}$, in which $v_{i}$ is the DL model parameters of car $i$.
For networks with the aggregation protocol (\ref{fjfjfjfjkkleie}), $\bar{v} =\frac{1}{| N_i |} \sum_{i} v_{i}$ is the corresponding converged value.
Mentioned that the convergence property is available while the car platoon is fully connected, i.e., $\lambda_2>0$, where $\lambda_{1} \leq \lambda_{2} \leq \cdots \lambda_{max}$ express the eigenvalues of $L$ in ascending order.

Moreover, the eigenvalue of the Laplacian matrix for the flying car network fully determines the distributed FL convergence.
We entitle $\epsilon = v(t) - \bar{v}(t)$ as the deviation of $v$ and the converged value $\bar{v}$.
$\lambda_2$ confines the convergence rate of the shared DL model parameters in the distributed FL. There is \cite{Saber1333204},
\begin{equation}
\begin{split}
\begin{aligned}
||\epsilon(t)|| \le ||\epsilon(0)|| e^{- \lambda_2 (L) t},
\end{aligned}
\end{split}
\label{gfiuofuyuuhyhk}
\end{equation}

\noindent where the convergence deviation $||\epsilon(t)||$ of the shared model parameters exponentially diminishes with time $t$, specified by $\lambda_2$. 
The larger $\lambda_2$ incurs the faster parameter convergence speed in FL.
But by definition $\lambda_2 \le \lambda_{max}$, increasing $\lambda_2$ also results in a growth of $\lambda_{max}$.

Unfortunately, according to Eq.~(\ref{sdajhgfduireawuihfcvhj}), the rise of $\lambda_{max}$ will make the tolerable delay $\tau_{max}$ become rigorous. 
That is, it puts forward stringent requirements for communications. 
The main purpose of distributed FL network design is to increase $\lambda_2$ while trying to prevent the increase in $\lambda_{max}$. 
Therefore, the topology optimization problem is formed as,

\begin{equation}
\begin{split}
\begin{aligned}
\textbf{P1:} \qquad &\min\limits_{\{L(G), \tau_{max}\}}  \lambda_{max}(L) - \eta \lambda_2(L) \\
\text { s.t. } \ \ \ \ \
&\text{(5a):} \ \tau_{max} \le \frac{\pi}{2 \lambda_{max}},
\end{aligned}
\end{split}
\label{erwqeiuidteryoa}
\end{equation}

\noindent where constraint (5a) is the transmission delay required to guarantee the distributed FL convergence.
Denote by $\eta$ the weight of $\lambda_2$.
$\tau_{max}=\max \lceil \frac{ M_{ij} }{ R_{ij} } \rceil $ represents the tolerable transmission delay of all links where $M_{ij} $ is the shared DL model size and $R_{ij}$ is the transmission rate of link $(i,j)$. 

Besides, the communication topology structure specifies $\lambda_{max}(L)$ and $\lambda_2(L)$ of the Laplacian matrix $L(G)$. 
Search the optimal communication topology $G$ through a traversal method to minimize the $\lambda_{max}(L) - \eta \lambda_2(L)$ is impractical and expensive.
A viable solution is to split the optimization goal into the minimizing $\lambda_{max}$ part and the maximizing $\lambda_2$ part.
Hereafter, we explore the topology controls for $\lambda_{max}$ and $ \lambda_2$, respectively.

\subsubsection{\textbf{Topology control for minimizing $\lambda_{max}$}}
based on the Gersgorin disk theorem \cite{horn_johnson_2012}, all the eigenvalues $\bm{\lambda}(L)$ of $L$ are located in the disk,
\begin{equation}
\begin{split}
\begin{aligned}
D(G) = \{ z \in \bm{\lambda}(L) \ | \ \|z-d_{max}(G)\| \le d_{max}(G) \}.
\end{aligned}
\end{split}
\label{stututuiweqfdfsguydfuoyfga}
\end{equation}

\noindent It implies the largest eigenvalue $\lambda_{max}\le 2d_{max}(G)$ where $d_{max}(G)$ is the largest vertex degree of network $G$.
Thus, reducing $d_{max}(G)$ can shrink the largest eigenvalue $\lambda_{max}$ of $L(G)$.
Moreover, the work in \cite{Nica2016ABI} proved that $\lambda_{max} = 2d_{max}(G)$ if and only if the graph is regular and bipartite.
A graph is regarded as a bipartite if its vertex set can be partitioned into two subsets $V_0$ and $V_1$, where no two points belonging to the same subset are connected by an edge, as shown in Fig~\ref{bipartite}.

\begin{figure}[h]
\centering
     \includegraphics[width=.28\textwidth]{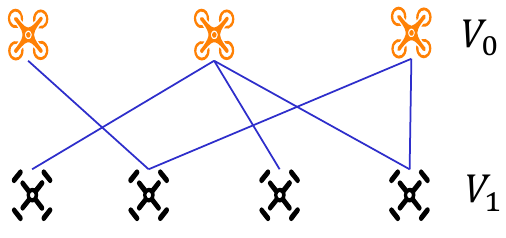} 
     \caption{Demonstration of the bipartite.} 
\label{bipartite}
\end{figure}

According to the bipartite graph theorem in the work \cite{Nica2016ABI}, a graph is bipartite if and only if it holds no cycles with an odd length. 
Alternatively, any path with an odd length in the bipartite has its endpoints on the different sides of the bi-partition. 
Consequently, we can construct an odd cycle to violate the bipartite form, which avoids $\lambda_{max}$ attaining the maximum.
Meanwhile, $d_{max}(G)$ of the network should be suppressed to reduce the upper bound of $\lambda_{max}$.

\subsubsection{\textbf{Topology control for maximizing $\lambda_{2}$}}
on account of Cheeger's inequality \cite{Nica2016ABI}, the relation of $\lambda_2$ to conductance $\Phi(G)$ is given as,
\begin{equation}
\begin{split}
\begin{aligned}
\frac{1}{2} \lambda_2 \le \Phi(G) \le \sqrt{2\lambda_2}.
\end{aligned}
\end{split}
\label{qbbbbbbbbbbg}
\end{equation}

\noindent Based on Eq.~(\ref{qbbbbbbbbbbg}), $\lambda_2$ increases with $\Phi(G)$ that is the conductance of network $G$ defined as,
\begin{equation}
\begin{split}
\begin{aligned}
\Phi(G) = \frac{\|\partial V_i\|}{\min\{d(V_i), d(V_i^c)\}}.
\end{aligned}
\end{split}
\label{tehpiorethierfgj}
\end{equation}
 
\noindent For a graph's non-empty vertex subset $V_i$, the boundary $\partial V_i$ is the set of links that connects $V_i$ to its complement $V_i^c$.
$\|\partial V_i\|$ indicates the element number of set $\partial V_i$.
Noted that $\|\partial V_i\|$ equals $\|\partial V_i^c\|$ since $V_i$ and its complement $V_i^c$ have the same boundary.
$d(V_i)$ and $d(V_i^c)$ refer to the sum of all vertex degrees of $V_i$ and $V_i^c$, respectively.
The intuitive purpose of the definition $\Phi(G)$ is to connect the most distant nodes as much as possible to increase the numbers of boundary $\partial V_i$.

\begin{figure}[h]
\centering
     \includegraphics[width=.28\textwidth]{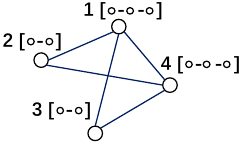} 
     \caption{Vertex degree structure.} 
\label{example_degree_structure}
\end{figure}
Instead of just counting the vertex degree, we explore the neighbor structure of a vertex to analyze the topological factor impacting the conductance $\Phi(G)$.
The neighbor structure enriches the regional information offered by vertex degrees. Given a graph, the neighbor structure has all vertices and edges inherited from the ambient subgraph of the investigated node.

As an illustration of this idea, the graph of Fig.~\ref{example_degree_structure} has two types of neighbor structures, $\Circle - \Circle$ and $\Circle - \Circle - \Circle$.
We named the form $\Circle - \Circle \dots$  as the connected neighbor structure and $\Circle \ \ \Circle \dots$ as the singleton neighbor structure. 
Hereafter, we turn to reveal a topological criterion for maximizing $\lambda_2$ based on the concept of the singleton neighbor structure.
\begin{theorem} 
\label{1111xxx}
Given graph $G$, the conductance $\Phi(G)$ attains maximum as all vertices exhibit the singleton neighbor structure, i.e., $\Circle \ \ \Circle \dots$.
\end{theorem} 

\begin{proof}
we assume that a vertex in graph $G$ with the maximum $\Phi(G)$ has the connected neighbor structure $\Circle - \Circle \dots$.
According to Eq.~(\ref{tehpiorethierfgj}), the conductance rises with $\|\partial V_i\|$ and reduces with $d(V_i)$ and $d(V_i^c)$, wherein $\|\partial V_i\|$ is the number of links connecting the sub-graph $V_i$ to its complement.
$d(V_i)$ is the sum of the vertex degree in the subset $V_i$.

Supposing a node $i$ with its connected neighbor structure is enclosed in subset $V_i$. Then, we assign all neighbors of the node into the complement set $V_i^c$. 
This operation increases the link number of the boundary $\|\partial V_i\|$ while producing the complement $V_i^c$ with redundant intra-links. 
This is because the connected neighbor structure of node $i$ proves at least one link between the neighbors of node $i$ in $V_i^c$.

Once we cut off the intra-links, the sum of the vertex degree $d(V_i^c)$ in $V_i^c$ will decrease.
As a result, based on Eq.~(\ref{tehpiorethierfgj}), the yielded conductance $\Phi'$ by cutting links is larger than the initial $\Phi$ with a connected neighbor structure.
It conflicts with the assumption that the maximum $\Phi$ holds a connected neighbor structure. 
Therefore, the neighbor structure of all nodes in graph $G$ with the maximum conductance $\Phi(G)$ must be the singleton, i.e., $\Circle \ \ \Circle \dots$.
\end{proof}

In summary, we propose topological criteria to optimize distributed FL implementation.
\begin{lemma} 
\label{222yyyy}
The topological criteria for distributed FL are given as,
\begin{enumerate}[label=\alph*).]
\item The maximum vertex degree $d_{max}$ should be lessened as much as possible.
\item The topology should have a cycle of odd length.
\item Link the longest distant node pair.
\item All vertices preserve the singleton neighbor structure. 
\end{enumerate}

\noindent Following the above topological criteria to generate a network can minimize the objective of $P1$, i.e., $\left[\lambda_{max}(L) - \eta \lambda_2(L)\right]$. 
\end{lemma} 



Our proposed topology control scheme is devoted to ensure shared parameter convergence and accelerate convergent speed. 
Without the topology control scheme, the shared parameters cannot attain the final converged value upon specific transmission constraints.
As we specify the communication network for distributed FL, the acceptable maximum communication delay $\tau_{max}$ between vehicles is bounded, i.e., $\tau_{max} \le \frac{\pi}{2 \lambda_{max}}$.
This tolerable transmission delay determines the point-to-point communication rate.
Subsequently, we leverage the RIS link modification ability to remold the communication network based on the topological criteria and accommodate the transmission rate of each communication link.




\subsection{RIS-empowered Topology Control}
This section proposes a RIS-based topology control for efficient distributed FL in UAM.
Each reflective element of RIS can program the phase of incident signals.
By appropriately shifting each element's phase, the incident signals are constructively added at the points of interest and deconstructively reduced at the point of aversive.
In contrast to retrospectives, whose RIS primary function concentrates on physical transmission power or attaining secure communication, we originally operate the RIS functions in the network layer to tailor the network topology for fast convergence of distributed FL.

Denote by $\mathbf{\mathbf{\phi}} = diag\{e^{j\phi_{1}}, e^{j\phi_{2}}, \dots e^{j\phi_{M}}\}$ the phase-shift matrix of the RIS, where $e^{j\phi_m}\in [0, 2\pi), \forall m \in \mathcal{M} $ represents the phase-shift of the $m$th RIS reflective element. 
$\mathcal{M} \triangleq \{1,2,\dots M\}$ is the set of indices between $1$ and $M$. 
And the channel state information (CSI) of communicating vehicles is assumed to be perfectly known on the RIS side.
The RIS controller can manipulate the phase-shift matrix $\mathbf{\phi}$ according to the collected CSI.

This paper exploits the Time Division Multiple Access (TDMA) scheme in flying car communication. 
The total bandwidth allocated to the platoon is divided into $U$ resource blocks (RBs), each of $W_b$ MB bandwidth. 
And the service periods of each vehicle are equal.
The per service period is further divided into a set of identical mini-time slots, denoted by $T_s \triangleq \{t_s, 2t_s,..., Kt_s\}$, where the duration of each mini-time slot is marked by $t_s$. 
$K$ is the number of the mini-time slot contained in a service period.

In this context, $\mathbf{h}_{\mathrm{car}, \mathrm{RIS}} \in \mathbb{C}^{N \times 1}$ represents the channel coefficients between a car and the RIS elements. 
Moreover, for all $m \in M$, let $\mathbf{h}_{i,m} \in \mathbb{C}$ and $\mathbf{h}_{\mathrm{RIS},m} \in \mathbb{C}^{N \times 1}$ be the channel coefficients of direct links from car $i$ and RIS to the $m$th car, respectively. 
Additionally, each communication link in the network is assumed to have a quasi-static flat-fading Rayleigh channel.
${s_i}$ is the transmitted data symbol for the $i$th car. 
Taking both the direct link and RIS cascaded links into account, we can give the received signal at the $m$th car as follows,
\begin{equation}
    {y_{im}} = \sqrt p ({{\mathbf{h}}_{i,m}} + {{\mathbf{h}}_{{\text{RIS}},m}}^H\phi {{\mathbf{h}}_{{i},{\text{RIS}}}}){s_i} + {z_m}, \  \forall m \in M,
    \label{wsdfheuirfhu3434289}
\end{equation}

\noindent where ${z_i}$ is white Gaussian noise and ${p}$ is the transmission power of flying cars. 
Due to the time division access, there is no interference at a receiver.
Consequently, the available transmission rate between $i$th car and $m$th car is,
\begin{equation}
    R_{im} = B{\log _2}\left(1 + \frac{{p\|{{\mathbf{h}}_{i,m}} + {{\mathbf{h}}_{{\text{RIS}},m}}^H\phi {{\mathbf{h}}_{{i},{\text{RIS}}}}\|^2}}{{{\sigma ^2}}}\right),
\label{shannel}
\end{equation}

\noindent in which, ${{\sigma ^2}}$ denotes the power of Gaussian white noise. 
The received signal experiences three channels: the direct channel between the transmitter and receiver, the channel between the transmitter and RIS, and the channel between RIS and the receiver. 
We can achieve the desired transmission rates by altering the RIS phase shift matrix.
Depending on the RIS mechanism, this paper explores signal construction and deconstruction to revise the communication topology for the distributed FL.

\subsubsection{Signal Construction} 
according to Eq.~(\ref{shannel}), we can increase the RIS channel component ${{{\mathbf{h}}_{{\text{RIS}},m}}^H\phi {{\mathbf{h}}_{{\text{car}},{\text{RIS}}}}}$ through adjusting $\phi$.
This operation can construct the signal power at the receiver, enhancing the transmission rate and bypassing building obstruction.

\subsubsection{Signal Deconstruction}
however, not all communication links need to enhance.
The link-adding procedure will consume more bandwidth resources and stringent the transmission delay requirements for distributed FL convergence since the link-adding operation increases $\lambda_{max}$.
However, the vehicular network may not fulfill the stringent delay requirement due to the limited bandwidth.
Hence, we must eliminate redundant links to lessen bandwidth resource consumption and relieve delay requirements.

According to Eq.~(\ref{shannel}), the transmission rate of the undesired link will approach $0$ as ${{\mathbf{h}}_{car,m}} + {{\mathbf{h}}_{{\text{RIS}},m}}^H\phi {{\mathbf{h}}_{{\text{car}},{\text{RIS}}}} = 0$.
Thus, the deconstructive RIS phase shift $\phi$ is,
\begin{equation}
\phi  =  - {({\mathbf{h}}_{{\text{RIS}},m}^H)^{ - 1}}{{\mathbf{h}}_{car,m}}{({{\mathbf{h}}_{{\text{car}},{\text{RIS}}}})^{ - 1}}
\label{dvoioewitubggogsiubh_v}
\end{equation}

The investigated flying car platoon is demonstrated in Fig.~\ref{ConveasDsaidhuywe}. 
The initial communication topology, as shown in Fig.~\ref{RIS_scenario1}, is molded in obstacles and path loss where car $6$ and $7$ are disconnected due to the building's obstructions. 
Car $3$ and $2$ cannot attach to distant vehicles due to the path loss.
Additionally, several redundant links in Fig.~\ref{RIS_scenario1} deteriorate the distributed FL performance while increasing bandwidth consumption and the largest eigenvalue $\lambda_{max}$, such as the link between car $3$ and $1$.
This link creates a connected neighbor structure of car $2$ that violates the criterion \textit{(d)} of Lemma \ref{222yyyy}.
Consequently, the link between car $3$ and $1$ should be deconstructive.
Meanwhile, according to the criterion \textit{(c)} of Lemma \ref{222yyyy}, we must construct the link (3,7) and link (2, 6).

\begin{figure}[h]
\centering
\subfloat[Initial topology.]{\label{RIS_scenario1}{\includegraphics[width=0.483\linewidth]{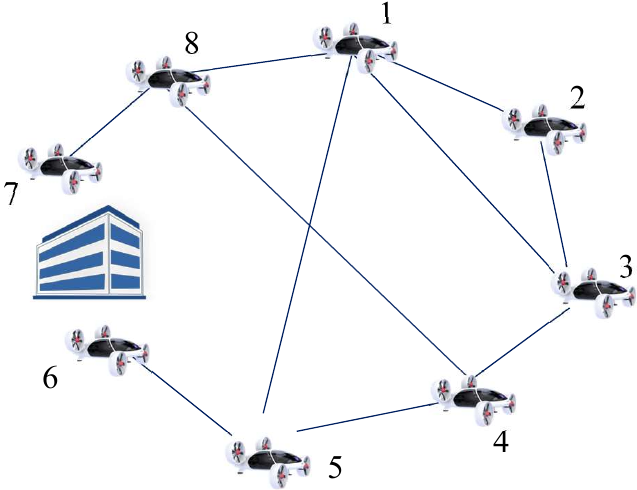}}} 
\hfill
\subfloat[Revised topology.]{\label{RIS_scenario2}{\includegraphics[width=0.47\linewidth]{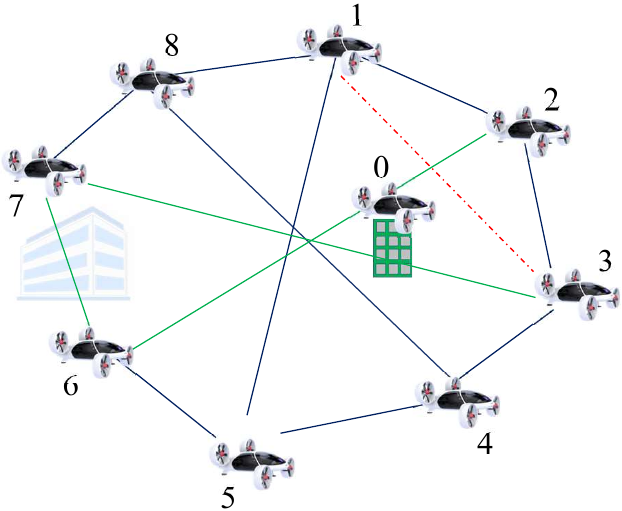}}}
\hfill
\caption{RIS-empowered topology reshaping. In the revised topology, the solid green and dashed red lines represent the newly constructive and deconstructive links, respectively.}
\label{ConveasDsaidhuywe}
\end{figure}

By leveraging RIS's link constructive and deconstructive abilities, we can control the RIS phase shift $\phi$ to reshape the initial topology, making it friendly to distributed FL.
According to the Lemma~\ref{222yyyy}, we revise the initial topology into Fig.~\ref{RIS_scenario2}, based on the topological criteria. 
The criteria-revised topology can boost distributed FL convergence and reduce bandwidth consumption while meeting the delay requirement.

In this paper, we regard the link transmission rate as determining the link connectivity.
Specifically, a link is disconnected as the transmission rate is less than $R_{lower}$.
Moreover, a link is available for distributed FL as the transmission rate over $R_{upper}$.
$\mathcal{N}_k$ is the set of vehicles intending to connect with car $k$ while ${\mathcal{N}'_k}$ is the set of vehicles towards disconnecting with car $k$ based on the topological criteria.
Therefore, the optimization solution of $P2$ specifies the transmission rates to tailor the communication topology.
\begin{equation}
\begin{split}
\begin{aligned} 
\textbf{P2:} \qquad &\max _{\mathbf{\phi}}  \sum_{k \in V(G)} \left(\sum_{i \in \mathcal{N}_k} R_{ki}-\sum_{j \in \mathcal{N}^{\prime}_k} R_{kj} \right) \\
\text { s.t. } \ \ \ \ \
&\text{(12a):}  \  \|\phi(i, j)\|=1 \\
&\text{(12b):}  \ R_{ki} \geq R_{\text {upper }}=\frac{2 \Lambda \lambda_{max}}{\pi}, i \in \mathcal{N}_k \\
&\text{(12c):}  \ R_{kj} \leq R_{\text {lower }}, j \in \mathcal{N}^{\prime}_k,
\end{aligned}
\end{split}
\label{111yweuvrto}
\end{equation}
\noindent where constraint (12a) ensures that the phase shift amplitude is equal to $1$. 
Constraint (12b) is the transmission rate required by the constructive links, in which $\Lambda$ is the network traffic volume, i.e., the volume of the shared DL model parameters.
This constraint implies that constructive link transmission delay must be equal or less than $\frac{\pi}{2\lambda_{max}}$.
Constraint (12c) is the transmission rate required by the deconstructive link.
We regard that the link is deconstructed well when the transmission rate is less than the $R_{lower}$ threshold. 

\subsection{DDPG-based RIS Optimization}
To solve $P2$, we first form the RIS-based link construction and deconstruction optimization as a Markov decision process, i.e., $S_\text{MDP} = (State, Action, Reward)$, then engages a deep deterministic policy gradient (DDPG) algorithm to explore actions. The MDP components are listed as,

\textbf{State} is comprised by the normalized transmission rate $S(t)$ of all links and the RIS phase shift matrix $\phi(t-1)$ at time slot $t-1$, written as $[S(t), \phi(t-1)]$.

\textbf{Action} is constituted of the RIS phase shift matrix $\phi\in [-\pi, \pi]$.
The element of the phase shift matrix is a complex number, which can be decomposed as the real and imaginary parts with the modulus equal to $1$.

\textbf{Reward} is configured as Eq.~(\ref{efwjlhgvwervyugasdohvouadf}) in each step.
The reward comprises the transmission bonus and the penalty for the delay violation.
For vehicle $i$ expected to construct the link with car $k$, i.e., $i\in \mathcal{N}_k$, we employ the transmission rate summation $\sum\limits_{i \in \mathcal{N}_k} {R_{ik}} $ as the transmission bonus.
When the transmission rate $R_{ik}$ is lower than the delay requirement of the distributed FL, there is a penalty of the link, i.e., $\sum\limits_{i \in \mathcal{N}_k} {[{\frac{2 \Lambda \lambda_{max}}{\pi}} -{R_{ik}}]^+}$, where $[x]^+$ represents $\max{(x,0)}$.
\begin{figure*}[h]
\centering
     \includegraphics[width=.88\textwidth]{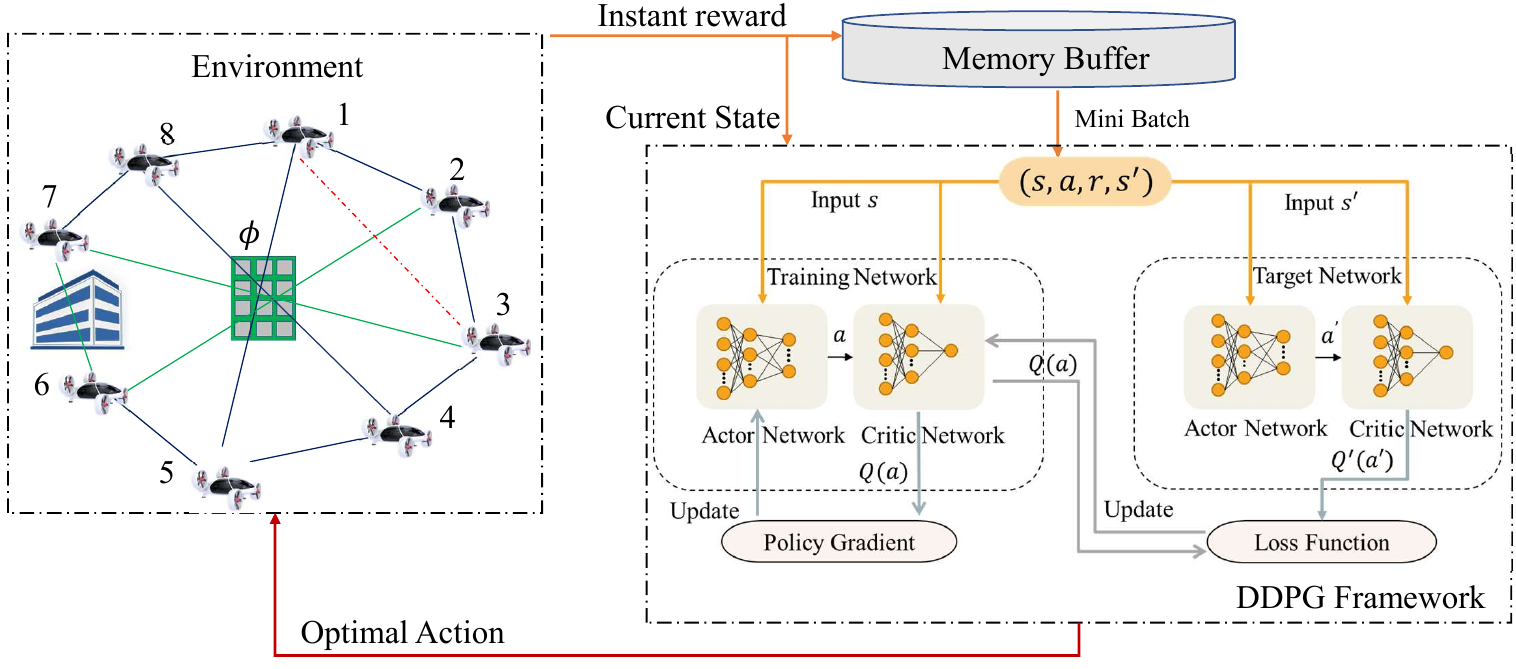} 
     \caption{DDPG-based RIS control framework.} 
\label{DDPG_Idea}
\end{figure*}

For vehicle $j$ expected to deconstruct the link with car $k$, i.e., $j \in \mathcal{N}'_k$, the transmission bonus is specified to $-\sum\limits_{j \in \mathcal{N}'_k} {R_{jk}}$.
In addition, the penalty of the deconstructive link is $\sum\limits_{j \in \mathcal{N}'_k}  [{R_{jk}} - {R_{lower}} ]^+$, indicating the deconstructed link's transmission rate should not exceed $R_{lower}$.
Therefore, we summarize the total reward of vehicle $k$ as,
\begin{equation}
\begin{split}
\begin{aligned} 
        r(t)_k &=  \sum\limits_{i \in \mathcal{N}_k} {R_{ik}} - \sum\limits_{j \in \mathcal{N}'_k} {{R_{jk}}}  - \\
        &  \gamma \left(\sum\limits_{i \in \mathcal{N}_k} {[{\frac{2 \Lambda \lambda_{max}}{\pi}} -{R_{ik}}]^+} +  \sum\limits_{j \in \mathcal{N}'_k} { [{R_{jk}} - {R_{lower}} ]^+} \right),  
\end{aligned}
\end{split}
\label{efwjlhgvwervyugasdohvouadf}
\end{equation}

\noindent where $\gamma$ is the weight of the penalty.
And the cumulative reward can be given at time $t$ as,
\begin{equation}
\begin{split}
\begin{aligned} 
        R(t)_k &= \sum\limits_{\tau=0}^{\infty} \beta(\tau)r_k(\tau+t-1), 
\end{aligned}
\end{split}
\label{fdgvirtgjrtgtr}
\end{equation}

\noindent in which $\beta$ is a discount factor for future rewards. The typical Q-function paired with the policy $\pi(s,a)$ is defined as $Q_\pi(s(t), a(t))$. Then, once the action $a$ is taken from state $s$, its Q-function is,
\begin{equation}
\begin{split}
\begin{aligned} 
      Q_\pi\left(s^{(t)}, a^{(t)}\right)=\mathcal{E}_\pi\left[R(t) \mid s(t)=s, a(t)=a \right].
\end{aligned}
\end{split}
\label{ujytkiofadvuwre5446}
\end{equation}

The proposed DDPG framework comprises actor and critic networks to substitute the state-action pair and value function in reinforcement learning, respectively.
The critic network generates Q value $Q(\theta^Q |s(t), a(t))$ to evaluate the current policy based on the instant rewards. 
Additionally, the actor-network updates the current policy $\pi(\theta^\mu|s, a)$ with hyper-parameters $\theta^\mu$. 
Wherein $\theta$ is updated by the gradient rule, $\theta(t+1) = \theta(t)-\zeta \nabla_{\theta} \mathcal{L}(\theta)$, in which $\zeta$ and $ \nabla_{\theta}$ are the learning rate and gradient of loss function $\mathcal{L}(\theta)$ respectively.

\begin{algorithm} 
	\caption{DDPG-based RIS Optimization} \label{Ag1}
	
    \textbf{Input}: normalized transmission rate of all links
    
    \textbf{Output}: optimal phase shift matrix $\phi^*$
    
	\textbf{Initialization}: given initialized experience replay memory; actor network ${\theta ^\mu }$; target actor network ${\theta ^{\mu'}}$, critic network ${\theta ^Q}$, target critic network ${\theta ^{Q'}}$, phase shift matrix $\phi$ \\
	\For{$eposide$ =0,1,2,...M-1} 
	{
		Observe initial state $state(0)$\\
		\For {$step$ = 0,1,2,...W-1} {
        Generate action $a(t+1) = \pi(\theta^\mu | s(t), a(t))$ from actor network

        Observe new state ${s(t + 1)}$ and reward $r(t + 1)$

        Save $({s(t)},{a(t)},{r(t + 1)},{s(t + 1)})$ in experience replay memory
   
    \If{experience replay memory is full} {
				Update actor and critic network
    
                 Update loss function according to Eq.~(\ref{c__fewugbeiuvbuebuasdccv})

                 Update the target network according to:
                 
                 $\begin{gathered}
  {\theta ^{Q'}} = \tau {\theta ^Q} + (1 - \tau ){\theta ^{Q'}} \hfill \\
  {\theta ^{\mu '}} = \tau {\theta ^\mu } + (1 - \tau ){\theta ^{\mu '}} \hfill \\ 
\end{gathered}$\\
			}
		}
	}				
\end{algorithm}

Due to the inability of the deep network to process imaginary values, this paper divides the input into real and imaginary parts, respectively, to run the deep network training regarding a real number way.
Therefore, as for the actor-network, the dimension of the input layer is $(2 M+n)\times 1$, in which $M$ depicts the number of RIS array elements. $2 M \times 1$ represents we encode the input as the real and image dimensions. 
$n$ is the number of transmitted signals. 
The actor output vector with the dimension of $2M\times 1$ represents the selected action at the current state. 
The structure of the critic network is the same as that of the action network, with only differences in the output layer, where the output vector is the fitted value of the Q-function with the $300\times 1$ dimension.

According to the framework of DDPG, there are two target networks, the target actor-network with parameters $\theta^{\mu'}$ and the target critic network with parameters $\theta^{Q'}$.
Moreover, the loss function $\mathcal{L}(\theta)$ is defined as the difference between $Q(\theta^Q |s(t), a(t))$ and $Q(\theta^{Q'} |s'(t), a'(t))$, which is expressed as,
\begin{equation}
\begin{split}
\begin{aligned} 
      \left.\ell(\theta)=\left(y-Q\left(\theta^{Q}\right) \mid s(t), a(t)\right)\right)^2.
\end{aligned}
\end{split}
\label{c__fewugbeiuvbuebuasdccv}
\end{equation}
\noindent where the target Q-value $y$ is,
\begin{equation}
\begin{split}
\begin{aligned}
     y=\left(r{(t+1)}+\alpha \max_{a'} Q\left(\theta^{Q'} \mid s', a'\right)\right) ,
\end{aligned}
\end{split}
\label{sdfguiyeruasdHIfghasuidhfewu_idew}
\end{equation}

\noindent where $\alpha$ denotes the weight coefficient. $Q(\theta^{Q'} |s'(t), a'(t))$ is the Q-value of target critic network.

The framework of the DDPG-based RIS control algorithm is illustrated in Fig.~\ref{DDPG_Idea}. 
According to the framework, the CSI of each car and the current RIS phase shift matrix are outputted to RIS.
DDPG derives the optimal action (the best RIS phase shift matrix) to optimize communication topology.

Alg.~\ref{Ag1} characterizes the proposed DDPG algorithm in detail, which can be well-trained offline.
Initially, a truncated normal distribution generates the actor and critic network parameters.
Alg.~\ref{Ag1} runs over $M$ episodes, each iterating $W$ maximum steps.
Each episode is terminated when the learning performance converges or reaches the maximum steps, i.e., $T$. 
The algorithm starts from the initial state $state(0) = [s(0), a(0)]$, in which $s(0)$ is the normalized transmission rate of all links at step $0$. 
The action $a(0)$ is configured as a matrix whose elements are complex numbers with a modulus equal to $1$.
The optimal action is yielded based on the current state and the instant reward.

 \section{Performance Evaluation}
This section conducted numerical simulations of the abovementioned scenarios and optimization algorithms. 
The RIS phase shift optimization scheme performs into each time slot since we employ the time division multiplexing on RIS communication. 
To evaluate the proposed topological criteria and the RIS-based distributed FL, we investigate the MobileNet-based multi-view learning using the proposed distributed FL framework, where MobileNet is a lightweight learning model designed to run on mobile devices proposed by Google \cite{HowardZCKWWAA17}.
Multi-view learning has become a hot topic during the past decades. 
It engages different views to bring diversity and redundancy to improve perception performance \cite{8998334Xiaodong}.

We investigate a scenario in which a car platoon is going to land on the street. The platoon must identify nearby ground vehicle categories to ensure landing safety and avoid collisions.
The involved multi-view dataset comprises the front, side, back, and vertical view pictures of various vehicles, a part of the dataset as shown in Fig~\ref{dashfgreiyaufgidsyugfiaydu}.

\begin{figure}[h]
\centering
     \includegraphics[width=.45\textwidth]{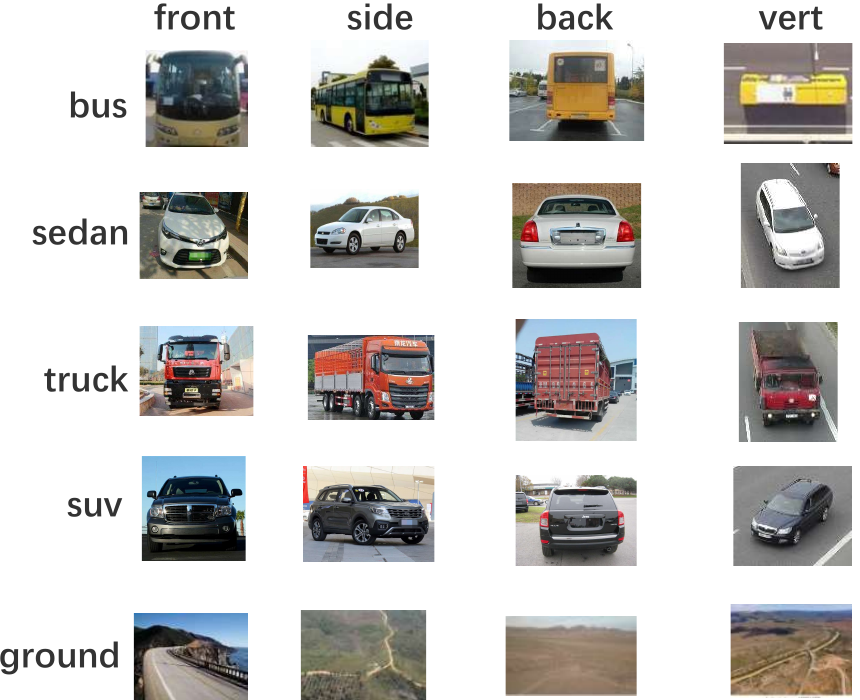} 
     \caption{Multi-view learning dataset. } 
\label{dashfgreiyaufgidsyugfiaydu}
\end{figure}

The dataset consists of $5000$ samples belonging to $4$ vehicle categories and $1$ background. 
Each category contains $4$ viewpoints, i.e., the front, side, back, and vertical. 
Furthermore, every viewpoint includes $250$ samples.
The flying car network before the RIS revising is molded by path loss and obstruction, as illustrated in Fig.~\ref{RIS_scenario1}. 
And the revised network is demonstrated in Fig.~\ref{RIS_scenario2}.
Cars will share their onboard DL model per 10 epochs of training.

To facilitate learning diversity, we divide the eight cars of Fig.~\ref{RIS_scenario2} into four parts. The car and node are equivalent in this paper.
Node $2$ and $6$ collect the side-view data. Node $3$ and node $7$ have front-view training samples. Node $4$ and node $8$ inspect the training samples from the vertical viewpoint. Node $5$ and node $1$ take account of the back-view samples to train their onboard models. 
Furthermore, the training dataset for each part is independently sampled from its corresponding view database, with a sampling probability of $70\%$. Thus, each part contains $875$ training samples. And the test set is the remaining $30\%$ of the view dataset.

\begin{figure}[h]
\centering
\subfloat[Convergence of distributed FL in the revised topology as shown in Fig.~\ref{RIS_scenario2}.  ]{\label{}{\includegraphics[width=0.99\linewidth]{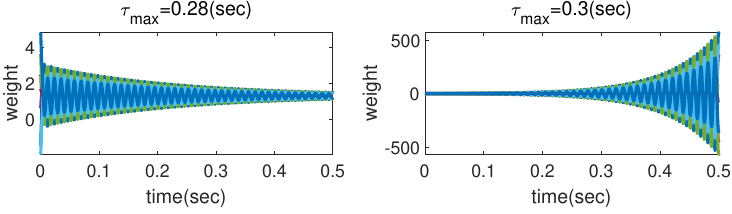}}} 
\hfill
\subfloat[Convergence of distributed FL in the initial topology as shown in Fig.~\ref{RIS_scenario1}. ]{\label{RT1_01}{\includegraphics[width=0.99\linewidth]{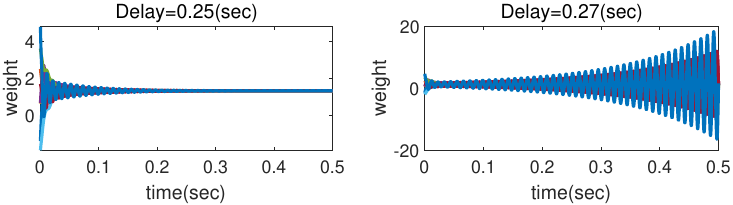}}}
\hfill 
\subfloat[Convergence of original FL in the star topology as shown in Fig.~\ref{TFL_scenario}. ]{\label{RT2_05}{\includegraphics[width=0.99\linewidth]{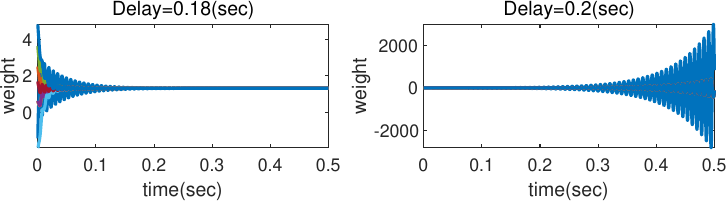}}}
\hfill
\caption{Convergence with different $\tau_{max}$ and topologies.}
\label{ConvergencgeasDUJ}
\end{figure}

Our proposed algorithm is run on a computer with NVIDIA GeForce RTX 2060 GPU and coded with PyTorch.
Some important simulation parameters are listed in Tab.~\ref{simulationargs},

 \begin{table}[!ht]
    \centering
    \caption{Simulation Parameters}
    \begin{tabular}{|l|l|l|}
    \hline
        \textbf{Description} & \textbf{Value} \\ \hline
       Line of sight loss coefficient  &  1.5 \\ \hline
       Non-line of sight loss coefficient & 4.0  \\ \hline
       Unit fading coefficient of path loss & 30 \\ \hline
       Bandwidth (Hz) & 3M \\ \hline
       Rice factor & 10 \\ \hline
       Vehicle transmission power(dbm) & 30  \\ \hline
       Number of RIS elements & 30-70  \\ \hline
       Discounted rate in DDPG &  0.9 \\ \hline
       Learning rate for training critic network in DDPG & 0.0001 \\ \hline
       Learning rate for training actor network in DDPG & 0.0001 \\ \hline
       Soft replacement in DDPG & 0.01 \\ \hline
       Buffer size for experience replay in DDPG & 10000 \\ \hline
       Batch size in DDPG & 32  \\ \hline
       Numbers of episodes in DDPG & 20000 \\ \hline
       Numbers of steps of each episode in DDPG & 100 \\ \hline
       Learning rate in MobileNet-based multi-view learning & 0.005 \\ \hline
       Batch size in MobileNet-based multi-view learning & 32  \\ \hline
        
    \end{tabular}
    \label{simulationargs}
\end{table}

\begin{figure}[h]
\centering
     \includegraphics[width=.28\textwidth]{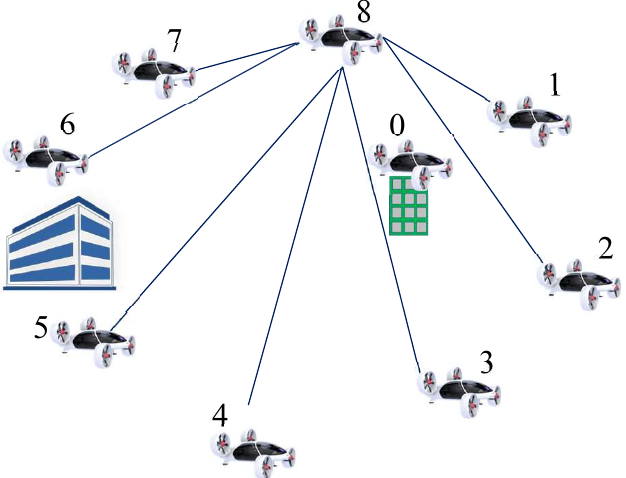} 
     \caption{Star topology for traditional FL integration.} 
\label{TFL_scenario}
\end{figure}

Fig.~\ref{ConvergencgeasDUJ} presents the convergence of distributed FL with different topologies.
The y-axis of Fig.~\ref{ConvergencgeasDUJ} is the average learnable weight parameters of a hidden node of the investigated neural network, i.e., MobileNet.
According to Eq.~(\ref{sdajhgfduireawuihfcvhj}), we can get the tolerable transmission delay for distributed FL convergence with revised topology, which equals $0.2901$.
Fig.~\ref{ConvergencgeasDUJ}a is the convergence of the distributed FL applied to multi-view learning with the topology Fig.~\ref{RIS_scenario2}.
The distributed FL converges with $\tau_{max} = 0.28s$.
However, the distributed FL is divergent with $\tau_{max} = 0.30s$ in the revised topology, i.e., Fig.~\ref{RIS_scenario2}.
\begin{figure*}
\centering
\subfloat{\label{}{\includegraphics[width=0.49\linewidth]{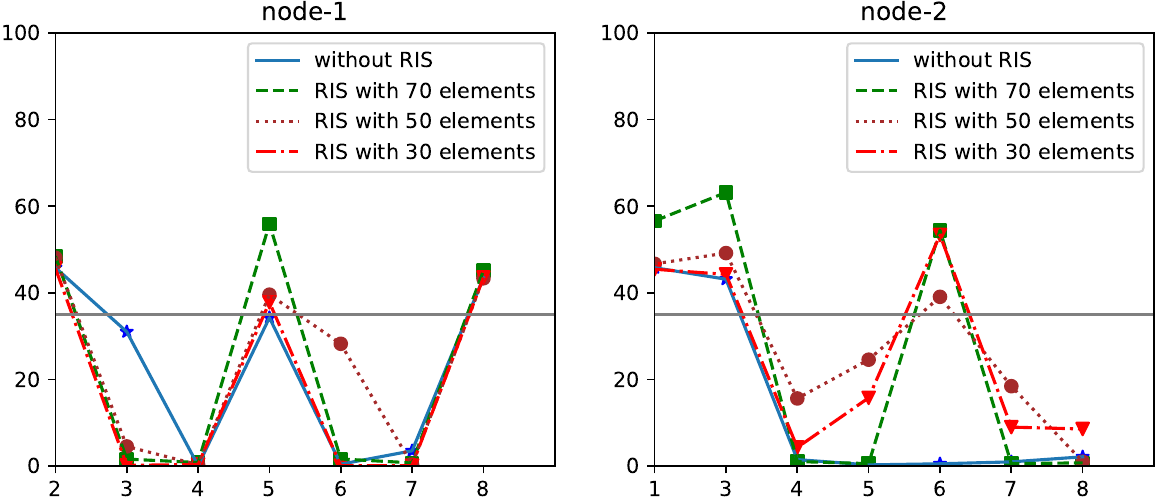}}} 
\subfloat{\label{RT1_01}{\includegraphics[width=0.49\linewidth]{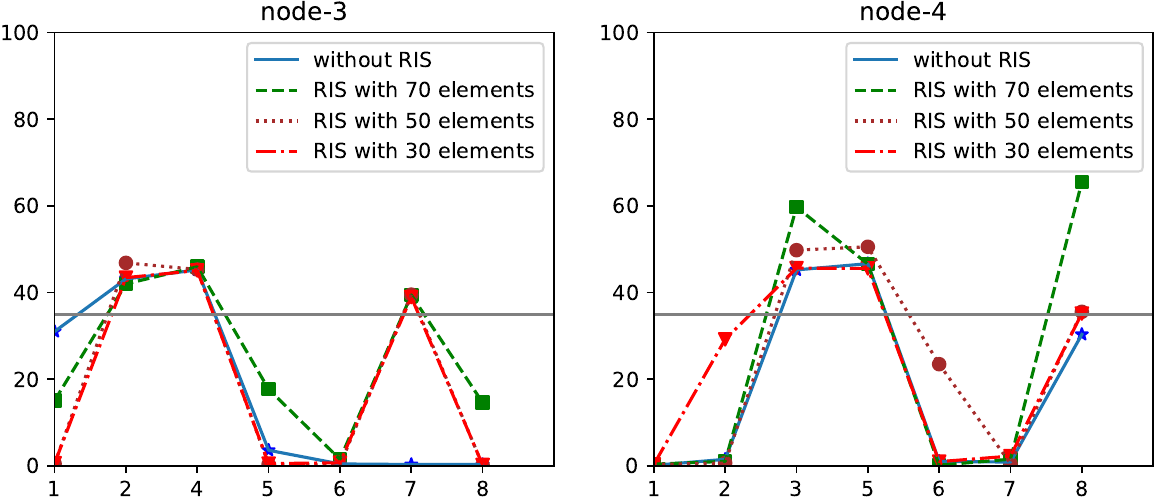}}}
\hfill 
\subfloat{\label{RT2_05}{\includegraphics[width=0.49\linewidth]{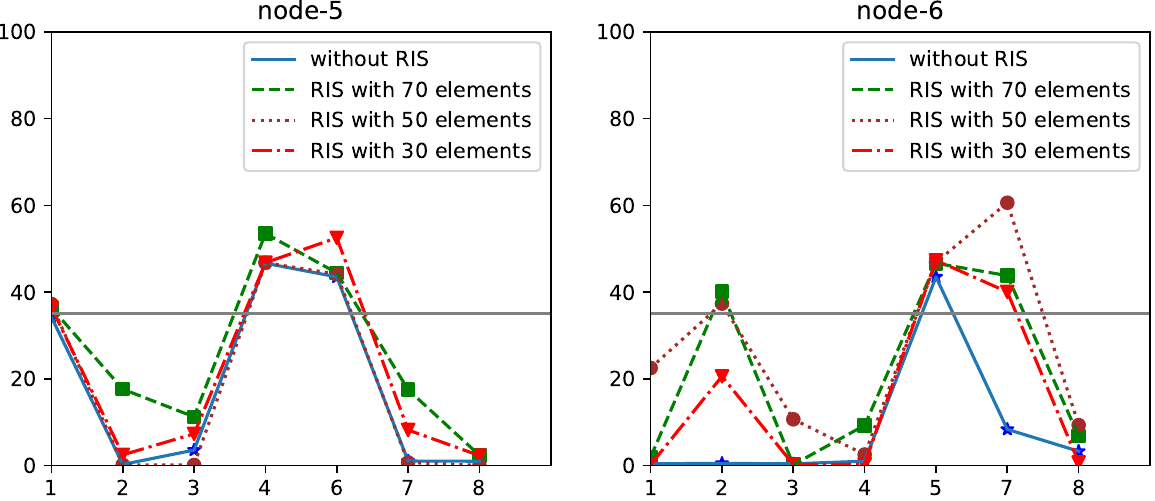}}}
\subfloat{\label{RT3_05}{\includegraphics[width=0.49\linewidth]{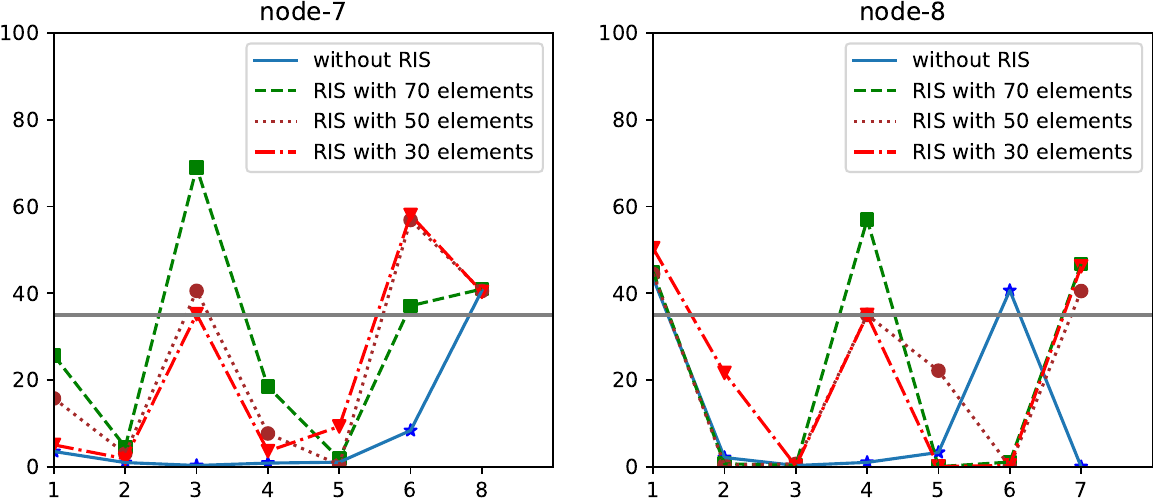}}} \hfill
\caption{RIS-controlled transmission rate of each node to remaining vertices, where the vertical axis represents the transmission rate, while the horizontal axis is the node ID. }
\label{dfasiughoeuwrhgdfsasdPP}
\end{figure*}

The tolerable transmission delay is $0.2682$ in the initial topology, i.e., Fig.~\ref{RIS_scenario1}.
Fig.~\ref{ConvergencgeasDUJ}b is verified initial topology convergence where the distributed FL is converged with $\tau_{max}=0.25$.
The distributed FL is divergent when $\tau_{max}=0.27$.
Fig.~\ref{ConvergencgeasDUJ}c investigates $\tau_{max}$ for the traditional FL convergence.
Generally, the traditional FL instructs a star topology to implement the DL model integration procedure, as shown in Fig.~\ref{TFL_scenario}. 
We get the tolerable delay $\tau_{max} = 0.1963$ of the star topology for the traditional/centralized FL convergence.
Fig.~\ref{ConvergencgeasDUJ}c confirms the delay convergence for star topology.
Therefore, it confirms the tolerable transmission delay growth obtained by topology optimization, reducing distributed algorithms' communication resource consumption.

Fig.~\ref{dfasiughoeuwrhgdfsasdPP} demonstrates the transmission rate from each node to all remaining vertices in the revised topology.
With the aid of RIS, we can construct and deconstruct the links by altering the transmission rate.
For instance, comparing the revised topology Fig.~\ref{RIS_scenario2} to the initial topology Fig.~\ref{RIS_scenario1}, node $1$ deconstructs the link with node $3$ and constructs the link with node $7$, in which the nodes and cars have the same connotation and are not distinguished.
We investigate the transmission rate of node $3$ in Fig.~\ref{dfasiughoeuwrhgdfsasdPP}.
Wherein the curve without RIS at the node $1$ has a high transmission rate of $30.98$ Mbps.
And the curve without RIS at the node $7$ is only $0.29$ Mbps.
In contrast, the curve with RIS can improve the transmission rate with the node $7$ and suppress the transmission rate with the node $1$ through the RIS phase shift optimization.
Applying the RIS can significantly enhance constructive and deconstructive performance. 
Further, the horizontal line in Fig.~\ref{dfasiughoeuwrhgdfsasdPP} represents the available transmission rate requirement for distributed FL convergence, i.e., $\frac{2 \Lambda \lambda_{max}}{\pi}$.
The transmission rate of each constructive link should exceed the available rate requirement.
As the experiments showed, the construction and deconstruction effects acquired by RIS are prominent and enhanced with the number of RIS elements.

We design a direct RIS control algorithm to demonstrate the effectiveness of DDPG algorithms. 
The channel matrix from a flying car $n$ to RIS is $h_{n,RIS}$, the channel matrix from RIS to a flying car $m$ is given as $h_{RIS,m}$, and the channel matrix between vehicles is $h_{n,m}$.
The phase shift matrix $X$ of RIS is the optimal variable to solve.

The direct RIS control algorithm only assesses the deconstruction link controls, which carries $h^H_{RIS,m} X  h_{n,RIS}+h_{n,m}=0$.
To simplify symbolic expression, we set $h^H_{RIS,m} = J$ and $h_{n,RIS} = B$.
And $h^H_{RIS,m} X h_{n,RIS}$ can be rewritten as,
\begin{equation}
   \left[ {\begin{array}{*{20}{c}}
  {{J_{11}}{X_{11}}}&{...}&{{J_{1m}}{X_{mm}}} \\ 
  {...}&{...}&{...} \\ 
  {{J_{n1}}{X_{11}}}&{...}&{{J_{nm}}{X_{mm}}} 
\end{array}} \right]\left[ {\begin{array}{*{20}{c}}
  {{B_{11}}}&{...}&{{B_{1p}}} \\ 
  {...}&{...}&{...} \\ 
  {{B_{m1}}}&{...}&{{B_{mp}}} 
\end{array}} \right]
\end{equation}
\noindent It indicates that the element in row $i$ and column $j$ of matrix $JXB$ holding
\begin{equation}
  {J_{i1}}{X_{11}}{B_{1j}} + {J_{i2}}{X_{22}}{B_{2j}} +  \cdots  + {J_{im}}{X_{mm}}{B_{mj}} = -h_{n,m}^{ij},
\end{equation}

\noindent where $h_{n,m}^{ij}$ is element in row $i$ and column $j$ of $h_{n,m}$.
These elements are scalar and can swap in order.
Therefore, we can element-wisely attain the optimal phase shift $X$.
Moreover, the RIS is passive reflection without amplifying or reducing the signal power.
To satisfy passive configurations, we scale each element of $X$ to a result with the modulus of $1$.
However, this operation will incur some errors.
\begin{figure}[h]
\centering
     \includegraphics[width=.42\textwidth]{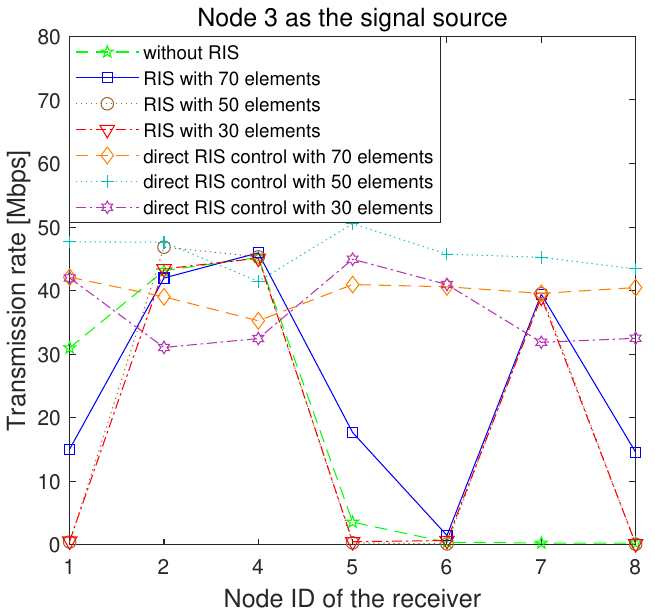} 
     \caption{RIS control algorithms comparison at node $3$.} 
\label{node3}
\end{figure}

Fig.~\ref{node3} illustrates the performance comparison between the DDPG RIS control and direct RIS control algorithms.
According to the topological criteria, node $3$ should construct the link with node $7$ while deconstructing the link with node $1$, as depicted in Fig.~\ref{RIS_scenario2}.
Wherein, node $3$ connects with node $2$, node $4$, and node $7$.
The DDPG RIS control algorithm enhances the transmission rate of these connected nodes and suppresses the transmission rate with node $1$. 
But the direct RIS control algorithm cannot fully suppresses the transmission rate of node $1$, node $5$, and other deconstructive links.
An interesting phenomenon is shown in Fig.~\ref{node3}. As the number of RIS elements increases, suppressing the transmission rate of deconstructive links is insufficient, like the link with node $1$. 
The possible reason is that both link construction and deconstruction are required in node $3$. 
The complexity grows with the number of RIS elements, leading to the underfitting training problem. 
We can solve this problem by introducing a larger DL model to train the RIS phase shift.

After packing the memory, we record the expense time of one training epoch averaged from $3000$ epochs. 
The time consumption with different RIS elements is illustrated in Fig.~\ref{time_show}. 
As the number of RIS elements gradually grows, the time consumption rises. The reason is that the actor-critic network's input dimension $(2M + n)\times 1$ increases with the number of RIS elements $M$. It yields growth in computing complexity.

\begin{figure}[h]
\centering
     \includegraphics[width=.4\textwidth]{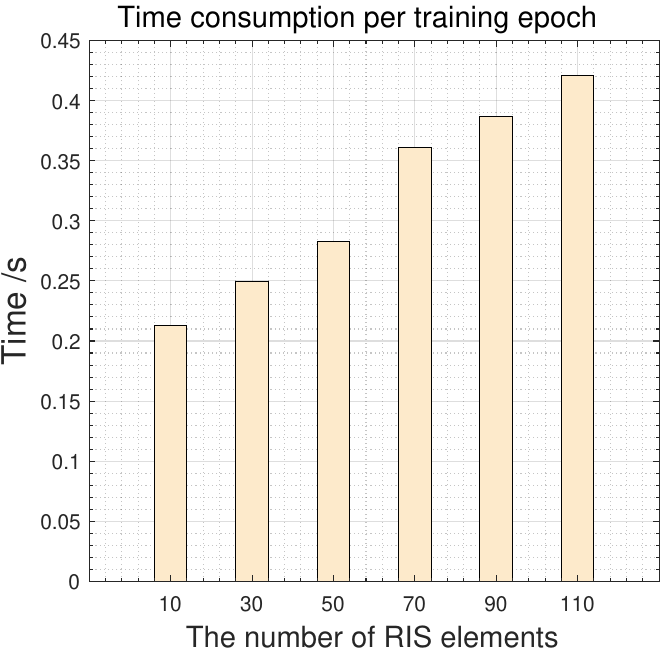} 
     \caption{Time consumption for one training epoch with different numbers of the RIS elements. } 
\label{time_show}
\end{figure}

Hereafter, we evaluate the perception accuracy of multi-view learning with different FL strategies and topologies in Fig.~\ref{PreceptionAcc}.
It demonstrates that the performance of distributed FL with revised topology is much higher than that of the initial topology.
The performance of distributed FL with the initial topology is even worse than the DL model without model sharing, as shown in Fig.~\ref{PreceptionAcc}. In that, the distributed FL with arbitrary topology does not ensure the convergence of the distributed parameters sharing.
The traditional FL, i.e., the centralized FL, has the highest accuracy at $30$ epochs due to the integrator's centralized DL model parameter aggregation.
However, according to Fig.~\ref{ConvergencgeasDUJ}, the centralized FL has the strictest tolerable transmission delay requirements and challenging communication resource assignments. 
Moreover, due to the centralized DL model aggregation, a network transmission bottleneck is prone at the integrator. A platoon cannot preserve an integrator in high dynamic UAM scenarios for the long term, which makes centralized FL unavailable.
\begin{figure}[h]
\centering
     \includegraphics[width=.45\textwidth]{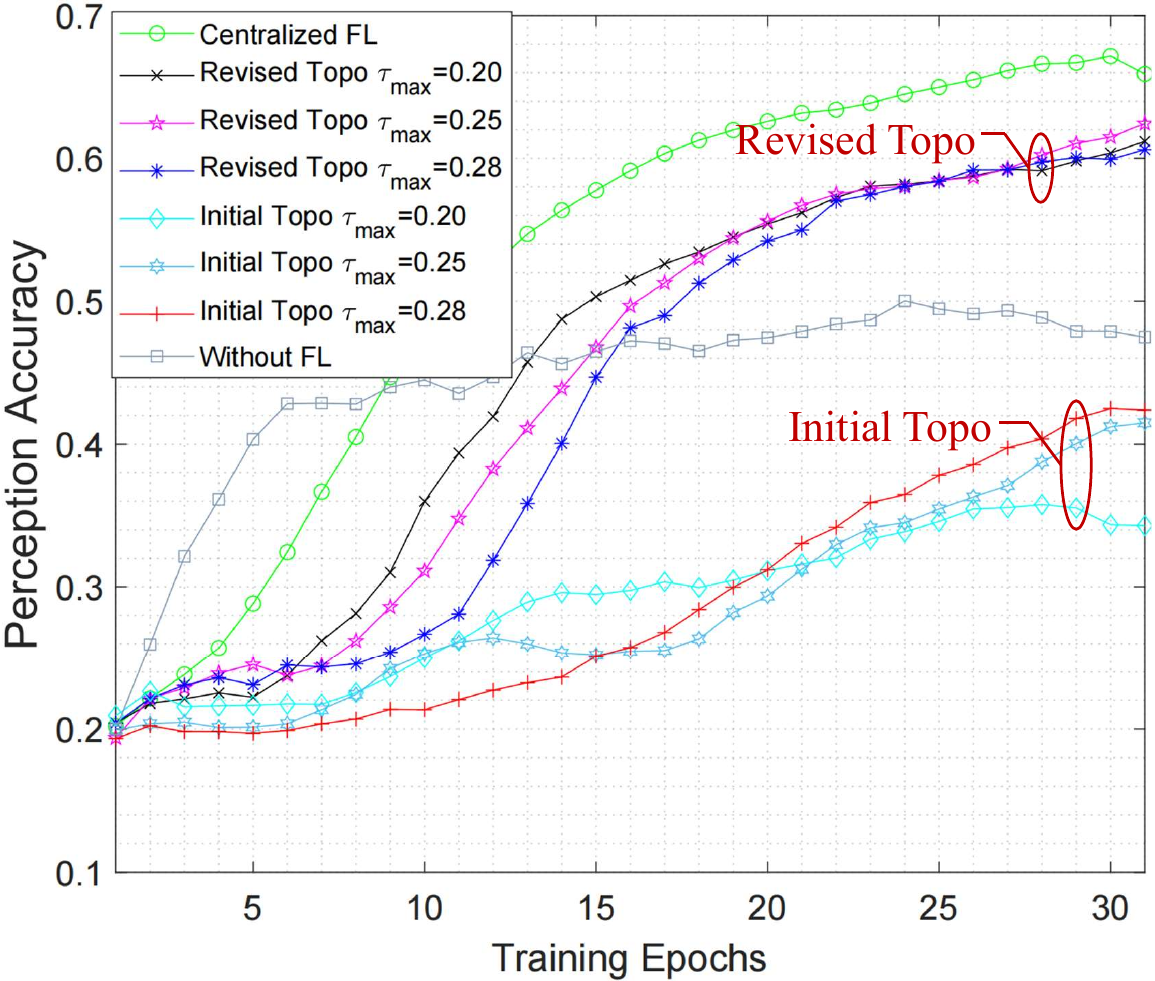} 
     \caption{FL performance with different FL fashions and topologies. } 
\label{PreceptionAcc}
\end{figure}

\begin{figure}[h]
\centering
     \includegraphics[width=.45\textwidth]{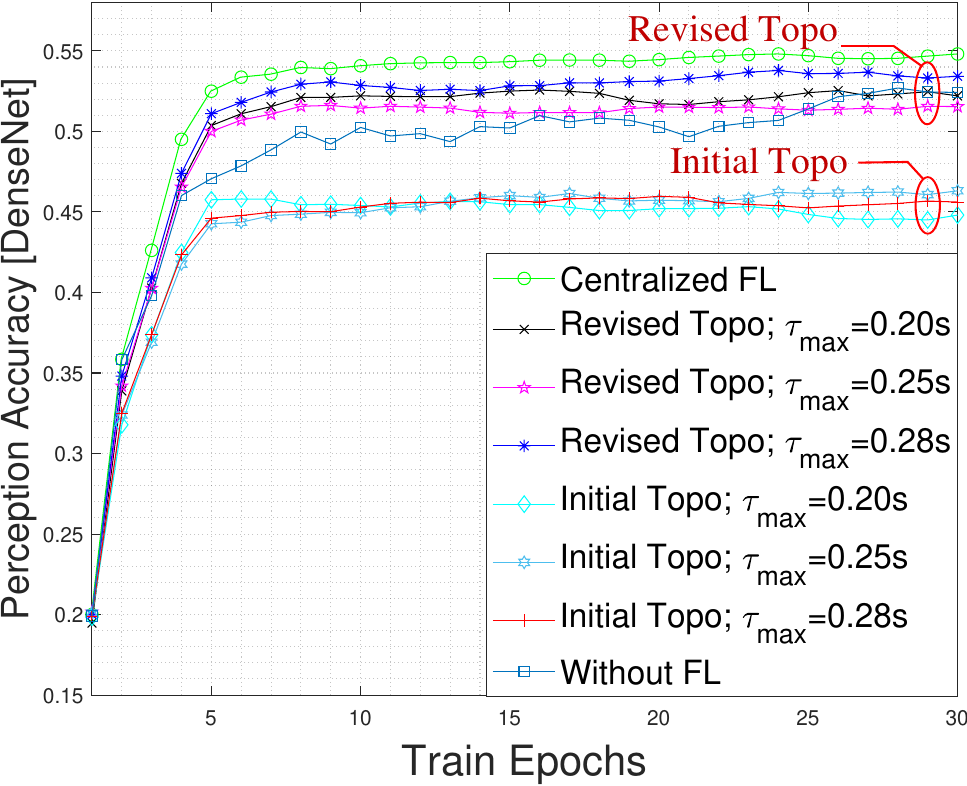} 
     \caption{Distributed FL performance with DenseNet. } 
\label{PreceptionAcc_DenseNet}
\end{figure}

To confirm the feasibility of the proposed distributed FL scheme, we have testified the other neural network, i.e., DenseNet.
Wherein, DenseNet is one type of convolutional neural network developed from ResNet. It is characterized by dense connections between different layers \cite{8099726}. This dense neural network has performed well in image recognition and computer vision tasks \cite{8331861Zhicheng}. It is also proper for resource-limited devices and embedded systems due to the compacted neural network structure \cite{9158008Yida}.

The performance of different distributed FLs with DenseNet is illustrated in Fig.~\ref{PreceptionAcc_DenseNet}, in which the traditional centralized FL and the individual DenseNet (Without FL) are the benchmarks to evaluate the diverse FL performances. The curve trends of the distributed FL with DenseNet are similar to those of MobileNet.
For instance, the performance of the revised topology is also higher than that of the initial topology. And the converged perception accuracy increases with $\tau_{max}$ reduction. 
It verifies the efficiency of our proposed distributed FL scheme in diverse neural networks.

However, after several iterations, the performance of individual DenseNet (Without FL) is proximate to the performance of the revised topology. The reason is that the learning structure of DenseNet is more complex than MobileNet's. The model volume of DenseNet is $27.3$MB while the volume of MobileNet is only $6.2$MB. Therefore, the individual DenseNet has qualified learning and functional fitting abilities. It gets comparable performance to the collaborative DenseNet with distributed FL.

Generally, our proposed distributed FL can improve the learning performance of individual DL models. However, as the DL model volume increases, the collaborative gain of distributed FL becomes smaller due to the adequate performance of the individual big-volume models. In addition, the communication overhead of distributed FL also increases with the shared model volume. It reveals that only lightweight DL models can fully leverage the benefits of distributed FL frameworks.

Fig.~\ref{Ring} illustrates the performance comparison between proposed distributed FL and the ring-based distributed learning.
Ring-based distributed learning is widely investigated in academics and industries. Baidu Company announced a ring-based Allreduce algorithm to accelerate the training of distributed DL models \cite{baiduddd}.
To implement a ring-based distributed learning, nodes are arranged in a unidirectional ring connected by communication links. Each node receives and forwards trained parameters to its neighbors at each training epoch.

\begin{figure}[h]
\centering
     \includegraphics[width=.45\textwidth]{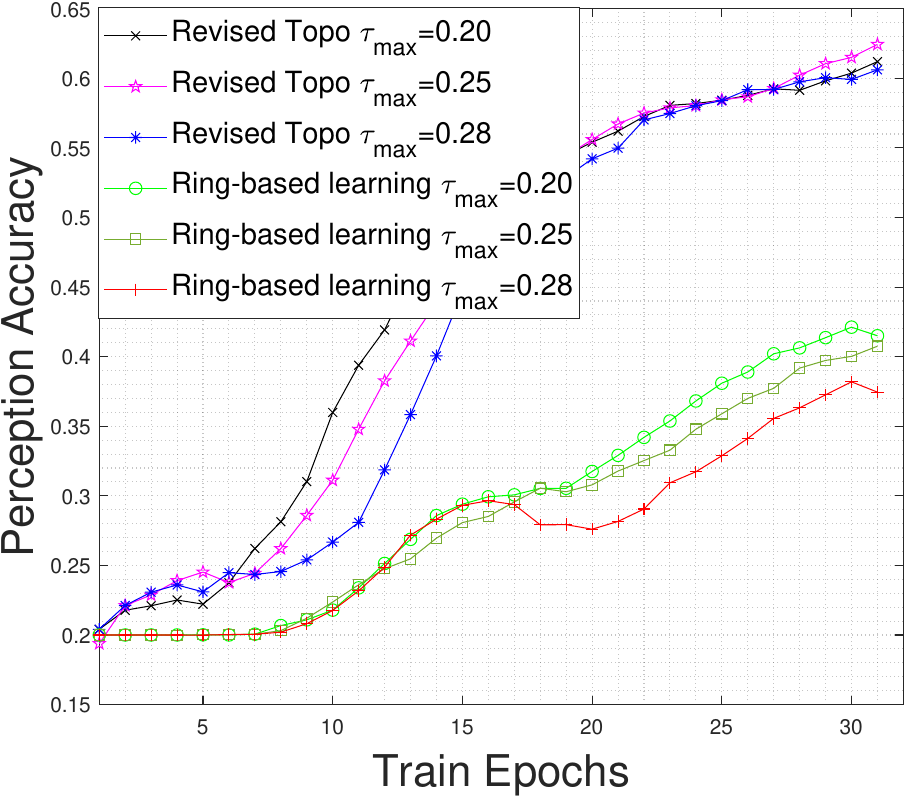} 
     \caption{Proposed distributed FL vs. Ring-based distributed learning. } 
\label{Ring}
\end{figure}

We can see that the perception accuracy of our proposed distributed FL with revised topology outperforms that of the ring-based distributed learning. in Fig.~\ref{Ring}.
The reason is that the second small eigenvalue $\lambda_2$ of the Laplacian matrix of the ring network is equal to 0.5858. It is much smaller than $\lambda_2$  of our proposed scheme which is 2. 
However, the larger $\lambda_2$ incurs the faster convergence speed of the parameter sharing. Thus, the ring-based distributed learning may not be eligible for the delay network. This performance comparison also verifies the efficiency of our proposed topological criteria.

Fig.~\ref{DFL_Delay} presents the distributed FL performance with different tolerable transmission delay $\tau_{max}$ in the revised topology. Perception accuracies of the curves with $\tau_{max} \le \frac{\pi}{2\lambda_{max}}$ are much higher than that of curves with $\tau_{max} > \frac{\pi}{2\lambda_{max}}$, where $\lambda_{max}$ is the largest eigenvalue of the revised topology Laplacian matrix. The reason is that the distributed FL convergence is guaranteed as the transmission delay does not exceed the tolerable delay $\tau_{max}$. With the transmission delay over $\tau_{max}$, the DL model sharing of distributed FL may not converge to an identical DL model, which deteriorates the DL model perception accuracy.

\begin{figure}[h]
\centering
     \includegraphics[width=.45\textwidth]{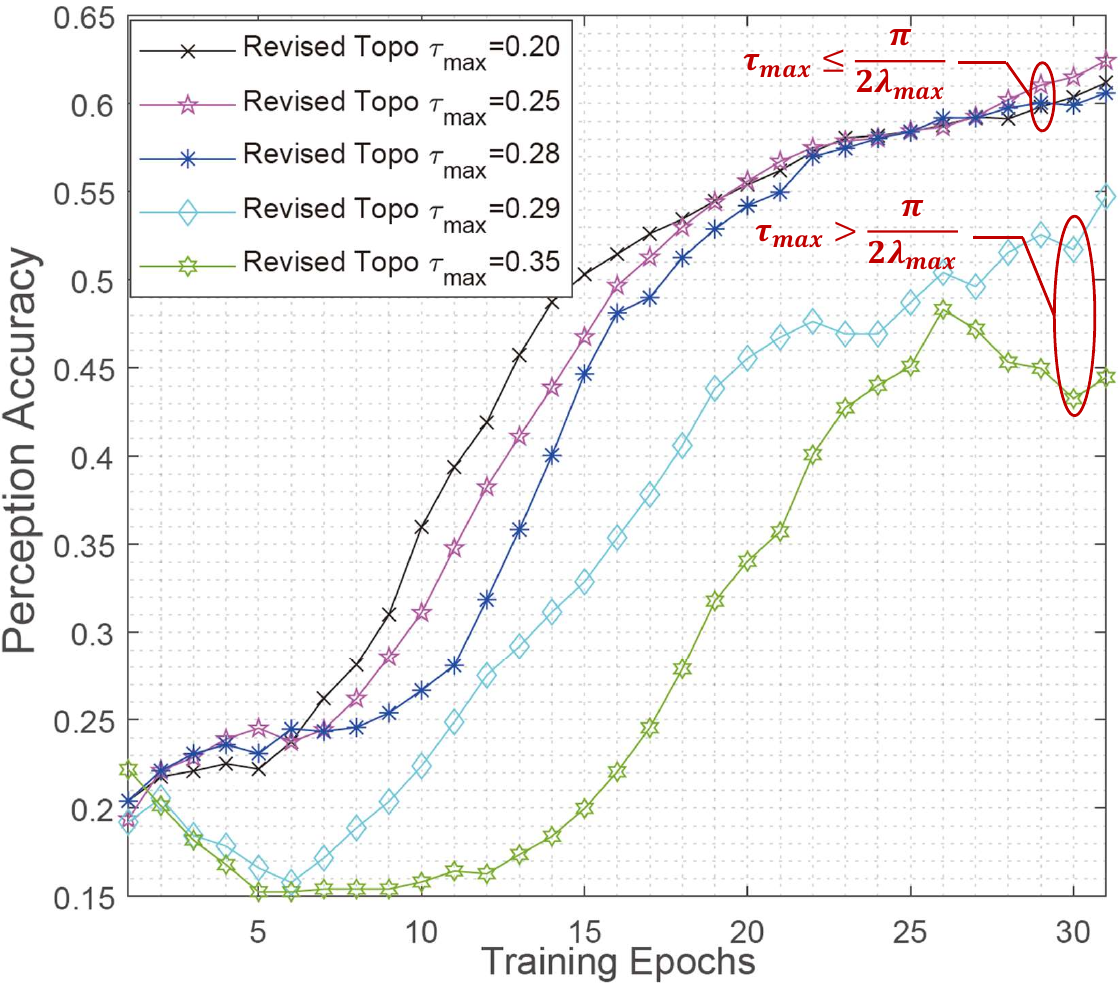} 
     \caption{Distributed FL performance with different tolerable transmission delay $\tau_{max}$ in the revised topology.  } 
\label{DFL_Delay}
\end{figure}

\section{Conclusion}
This paper proposed a RIS-empowered topological control scheme for distributed FL in UAM. It is applied to promote the onboard learning capacity of flying cars in the metropolis. With RIS's link construction and deconstruction abilities, flying cars can reshape their communication network based on the proposed topological criteria. These criteria minimize the largest eigenvalue and maximize the second smallest eigenvalue of the Laplacian matrix of the vehicular network. It facilitates distributed FL convergence and reduces communication overheads. Subsequently, we propose a DDPG-based RIS control scheme, optimizing the phase shift matrix of the RIS to accommodate the transmission rate.

Compared to previous work \cite{9110869Chongwen}, our proposed DDPG scheme leverages the RIS constructive and deconstructive abilities to reshape the communication topology for distributed parameter sharing instead of beamforming enhancement in the MISO context. Moreover, the investigated scenario of \cite{9416239Sixian} only contains one legitimate user and one eavesdropper. In contrast, we explored RIS link modification to strengthen the transmission rate of multiple constructive users and weaken that of multiple deconstructive users, simultaneously. Overall, most RIS retrospectives are dedicated to transmission enhancement in the physical layer. This paper originally leveraged the RIS functions in the network layer to tailor the network topology for fast convergence of distributed FL. This work provisions the possibility of utilizing distributed collaboration learning, which shows great potential for swarm intelligence in heterogeneous UAM systems.

In this paper, we only investigate the distributed FL in the network with a low topological change. When the network topology changes fast, it could impact the convergence of the distributed FL. Moreover, while the time consumption of topology change is less than the transmission delay, existing theories do not guarantee the convergence of distributed FL. In future work, we will explore the converged distributed FL in the network with fast topological change.


\appendices
\footnotesize
\bibliography{biblio}
\end{document}